\documentclass[conference]{IEEEtran}
\usepackage{amsfonts, amsmath, amsthm, amssymb}
\usepackage{mathtools, graphicx, epstopdf, tikz, standalone}
\usepackage{multirow}

\usepackage{atbegshi}
\IEEEoverridecommandlockouts
\newtheorem{lemma}{Lemma}
\newtheorem{prop}{Proposition}

\newtheorem{remark}{Remark}

\newcommand{\K}{\mathcal{K}}

\newcommand{\PP}{\mathbb{P}}
\newcommand{\E}{\mathbb{E}}

\allowdisplaybreaks
\usepackage{cite, graphicx, textcomp, multirow}
\begin{document}
	
\title{On the Design of Magnetic Resonant Coupling for Wireless Power Transfer in Multi-Coil Networks}	
\author{\IEEEauthorblockN{Eleni Demarchou, \IEEEmembership{Member, IEEE}, Constantinos Psomas, \IEEEmembership{Senior Member, IEEE}, and\\ Ioannis Krikidis, \IEEEmembership{Fellow, IEEE}\vspace{-3mm}}
\thanks{E. Demarchou, C. Psomas, and I. Krikidis are with the IRIDA Research Centre for Communication Technologies, Department of Electrical and Computer Engineering, University of Cyprus, Cyprus (e-mail: \{edemar01, psomas, krikidis\}@ucy.ac.cy).

This work was co-funded by the European Regional Development Fund and the Republic of Cyprus through the Research and Innovation Foundation, under the projects INFRASTRUCTURES/1216/0017 (IRIDA) and POST-DOC/0916/0256 (IMPULSE). This work has also received funding from the European Research Council (ERC) under the European Union's Horizon 2020 research and innovation programme (Grant agreement No. 819819).}}
\maketitle
	
\begin{abstract}
Wireless power transfer (WPT) is a promising technology for powering up distributed devices in machine type networks. Over the last decade magnetic resonant coupling (MRC) received significant interest from the research community, since it is suitable for realizing mid-range WPT. In this paper, we investigate the performance of a single cell MRC-WPT network with multiple receivers, each equipped with an electromagnetic coil and a load. We first consider pre-adjusted loads for the receivers and by taking into account spatial randomness, we derive the harvesting outage probability of a receiver; for both the strong and loosely coupling regions. Then, we develop a non-cooperative game for a fixed receiver topology, in order to acquire the optimal load which maximizes each receiver's harvested power. Throughout our work, we obtain insights for key design parameters and present numerical results which validate our analysis. 
\end{abstract}

\section{Introduction}
The flexibility of energy harvesting brought by wireless power transfer (WPT) has received significant attention from both industry and academia, over the recent years, and is a promising solution for the forthcoming machine type networks \cite{Alouini6G}. Wireless chargers for electric toothbrushes and cellular phones constitute of the most well-known applications realizing near-field WPT. Such applications are based on inductive coupling, which is only efficient for short-range energy harvesting, i.e., at the order of centimeters \cite{Barman}. As a remedy to extend the distance limitations, the authors in \cite{MIT}, performed an experiment implementing magnetic resonant coupling (MRC). The MRC approach is based on the primary concept of inductive coupling with both, the transmitter and the receiver, resonating at the same frequency \cite{Barman}. As a result, a strong magnetic coupling occurs, and non-radiative power can be transferred efficiently at longer distances \cite{MIT}. This is realized by an RLC circuit placed close to each electromagnetic (EM) coil \cite{MIT}, or by a variable capacitor alignment embedded in the electric circuit of each node\cite{compensator}. Then, resonance is achieved by tuning the corresponding elements. The efficiency of the MRC approach has attracted significant interests in the research community, since it is applicable for employing mid-range WPT \cite{Barman}.

A circuit model is provided in \cite{IC}, in order to capture the resonance coupling systems and it is shown that high resonant coupling is key for the power transfer efficiency. Aiming to maximize the information capacity of MRC communications, the authors in \cite{LEE}, investigate both the strongly and loosely coupled regions and define the optimal parameters which provide high information capacity. A system supported by a relay-acting coil is investigated in \cite{misalignment}, where the effect of lateral and angular misalignment is captured. In \cite{Rui}, the authors consider a WPT system with a single transmitter and multiple receivers and characterize the “near-far” fairness issue, which occurs from the distance-dependent mutual inductance with the transmitter. 
%
Furthermore, in \cite{Rui2} the optimal magnetic beamforming is derived for a multiple-input single-output setup, in order to maximize the delivered power at the receiver. The authors in \cite{MIMO}, consider a system with multiple coils at the receiver and multiple single-coil transmitters, and jointly optimize the transmitter source currents and receive coil selection by taking into account a minimum harvested power constraint. Coil selection is also studied in \cite{MIMO2}, where power leakage to unintended receivers is considered and a simultaneous magnetic field focusing and null-steering optimization problem is solved. A multiple-input multiple-output setup is investigated in \cite{Xu}, where limited feedback information on the magnetic channel is considered, and a random magnetic beamforming is proposed to improve the received power. 

In contrast to the existing literature, this letter focuses on the gap associated with the MRC-WPT network design from a system level perspective. While in this context, the research background is wide in the area of far-field WPT \cite{kWPTRelays}, the mid-range WPT network design is still in its infancy. In this work, we provide preliminary analysis for a single cell MRC-WPT network consisting of single-coil nodes. Specifically, \begin{itemize}
\item We consider spatial randomness for the receivers coils and analyze the performance of the typical receiver. By taking into account minimum harvested power requirements, we derive a tight upper bound for the outage probability in the strong coupling region, and provide a closed-form expression for the loosely coupling region.
\item By considering a fixed set of coils' profiles, we develop a non-cooperative game, in order to acquire the optimal load for each receiver and maximize the harvested power subject to the other receivers' profile.
\item Throughout our analysis, we obtain insights for several parameters which are essential for the MRC-WPT network design. Finally, we present numerical results to validate our analysis and show the gains brought by optimally adjusting the loads at the receivers.  
\end{itemize}  

\section{System model}

Consider a mid-range WPT system consisting of a single transmitter and a set $\K$ of $K$ receivers. The transmitter is equipped with an EM coil of self-inductance $L$, an internal resistance $R$, a capacitor $C$, and is connected to a sinusoidal voltage source $v_T$ providing an output power denoted by $P$. The receivers are equipped with identical EM coils of self-inductance $l$, internal resistance $r$, and a capacitor $c$, with the $i$-th receiver connected with an adjustable load $x_i$ \cite{Rui}. The network operates at an angular frequency $\omega\triangleq\frac{1}{\sqrt{LC}}$, and the receivers harvest power through the MRC-WPT principle i.e., $\frac{1}{\sqrt{lc}}=\frac{1}{\sqrt{LC}}$. The equivalent circuit model of the network is depicted in Fig. \ref{sm}. 

By applying the Kirchhoff's law, the harvested power at the $i$-th receiver can be expressed as \cite{Rui}
\begin{align}
p_i=P \frac{\omega^2 M_i^2 x_i (r+x_i)^{-2}}{R+\omega^2 \sum_{k=1}^{K}M_k^2 (r+x_k)^{-1}},\label{obj}
\end{align}
where 
$M_i$ denotes the mutual inductance between the transmitter and the $i$-th receiver. Let $e$ denote a constant defined by the physical characteristics of the coils, then the mutual inductance can be approximated by \cite{LEE}
\begin{equation}
	M_i = e \frac{I_i}{d_i^3},
\end{equation}
where $d_i$ denotes the distance between the $i$-th receiver and the transmitter and $I_i$ captures the impact of angular alignment with the transmitter and is given by \cite{LEE}
\begin{equation}
I_i= 2\sin(\theta_t)\sin(\theta_i)+\cos(\theta_t)\cos(\theta_i),\label{I}
\end{equation}
 with $\theta_t$ and $\theta_i$ corresponding to the transmitter's and $i$-th receiver's angle, respectively, between the coil radial direction and the axis between the transmitter's and receiver's coil centers. 
 
 { In the following we investigate two performance metrics; the harvesting outage probability and the harvested power. The former, derived in Section \ref{Soutage}, serves as a network-based evaluation metric whereas the latter is examined in Section \ref{optloads} as a receiver-centric metric.}
\begin{figure}[t]\centering	
	\includegraphics[width=0.7\linewidth]{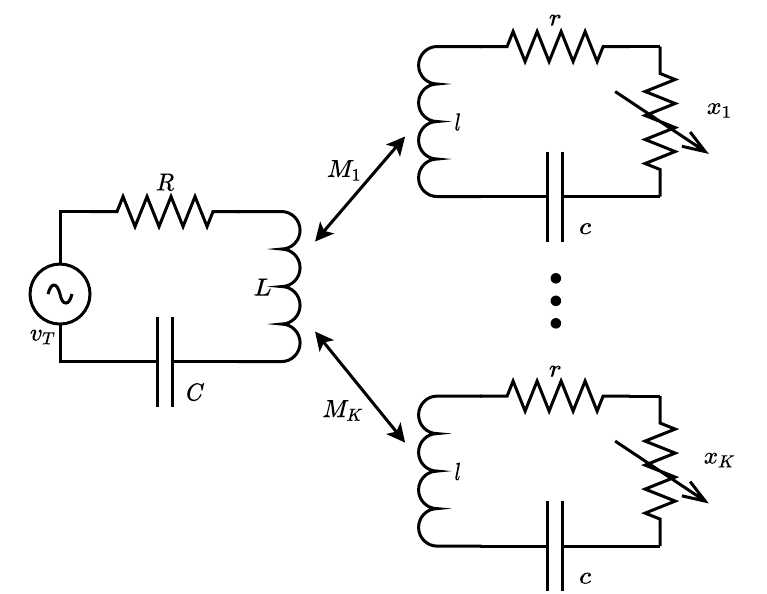}
	\vspace{-1mm}
	\caption{A single-cell MRC-WPT network with $K$ receivers.}\label{sm}
	\vspace{-2mm}
\end{figure}
\section{Harvesting Outage Probability Analysis}\label{Soutage}
We focus on the performance of a \textit{typical} receiver with fixed mutual inductance with the transmitter i.e., $M_0=e\frac{I_0}{d_0^3}$, where $I_0$ and $d_0$ denote the angular alignment and the distance between the typical receiver and the transmitter, respectively. { The other receivers are spatially distributed according to a homogeneous Poisson point process (PPP) $\Psi$ of density $\lambda$. We consider that the mutual inductance between the transmitter and any receiver located within a distance higher than $\rho$, is negligible. Therefore, we focus on a cell of radius $\rho$ centered at the transmitter \cite{kWPTRelays}.}
{{For the sake of simplicity, we consider that all the loads of the receivers are adjusted to $x_i\triangleq x$, while the impact of loads' diversity is considered in Section \ref{optloads}}}. {Furthermore, we take into account a minimum harvested power threshold $\tau$, which is required to satisfy the receiver's quality of service demands \cite{Rui}.} In the following we derive the outage probability of the typical receiver, which is defined as the probability that the harvested power is below the predefined threshold $\tau$ i.e., $\mathcal{P}_o(\tau)=\mathbb{P}\left[p_0<\tau\right]$, where $p_0$ is the harvested power at the typical receiver and is given by \eqref{obj}. 

To assist with our analysis, we let $S \triangleq \sum_{y_i \in \Psi} \frac{I_i^2}{d_i^{6}}$, where $y_i$ denotes the coordinates of the $i$-th receiver. Following the aforementioned, the harvested power at the typical receiver can be expressed by
\begin{equation}\label{u0}
p_0=P \frac{\omega^2 M_0^2 x (r+x)^{-2}}{R+\omega^2(r+x)^{-1}e^2\left(\frac{I_0^2}{d_0^6}+S\right)}.
\end{equation}
For the derivation of the outage probability we make use of the following lemma.
\begin{lemma}\label{lemma1}
	The moment generating function $\phi_S(t)$ of the random variable $S$ is evaluated by
	{\begin{equation}
	\phi_S(t)\approx \exp\left(\pi\lambda\left(\int_0^\rho\exp\left(tv^{-6}\right) v\,dv-\rho^2\right)\right).
 \end{equation}}
\end{lemma}

\begin{proof}
	See Appendix \ref{prf_lemma1}.
\end{proof}
Provided with Lemma \ref{lemma1}, we evaluate the outage probability directly with the help of the Gil-Pelaez inversion theorem \cite{GP}
\begin{align}
\PP(Y < y) = \frac{1}{2} - \frac{1}{\pi} \int_0^\infty \frac{1}{t} \Im\{\phi_Y(\jmath t) \exp(-\jmath t y)\} dt,
\end{align}
where $\jmath=\sqrt{-1}$, $\Im\{X\}$ gives the imaginary part of $X$ and $\phi_Y(jt)$ corresponds to the characteristic function of the random variable $Y$. The outage probability is then given in the following proposition.
\begin{prop}\label{out}
	The outage probability of the typical receiver is given by
	\begin{align}\label{mu}
	\mathcal{P}_o(\tau) = \frac{1}{2} + \frac{1}{\pi} \int_0^\infty \frac{1}{t} \Im\left\{\phi_S(\jmath t) \exp(-\jmath t \Lambda(\tau))\right\} dt,
	\end{align}
	where
	\begin{align}
	\Lambda(\tau) = \frac{I_0^2}{d_0^6}\frac{P x (r + x)^{-1} - \tau}{\tau}-\frac{R (r + x)}{\omega^2e^2},
	\end{align}
	if $P\geq\frac{\tau\left(r+x\right)}{x}\left(\frac{R(r+x)}{\omega^2M_0^2}+1\right)$, otherwise $\mathcal{P}_o(\tau)=1$.
\end{prop}
\begin{proof}
The result follows by solving the outage probability with respect to $S$ and by applying the Gil-Pelaez theorem \cite{GP}. Since $S>0$, it follows that $\mathcal{P}_o(\tau)=1$ when $\Lambda(\tau)<0$.
\end{proof}
For the special case $\lambda \to \infty$, the characteristic function above dominates the outage probability. Hence, since the imaginary part of a characteristic function simplifies to a sine function and as $\int_0^\infty \sin(at)/t = \pi/2$ for any $a \geq 0$, we have $\mathcal{P}_o (\tau)= 1$.  

We now turn our attention to the loosely coupled regime. That is, even though resonance is achieved, due to the path-loss attenuation or the angular misalignment, the mutual inductance with the transmitter is degraded such that $M_i\to0$, $\forall i \in \K$.  In this case, the effect of coupling is neglected at the transmitter's coil \cite[Sec. III]{LEE}, and the received power at the $i$-th receiver is given by \cite{Rui}
\begin{equation}
p_i=P\frac{\omega^2 M_i^2 x_i}{R(r+x_i)^2}. \label{s1}
\end{equation}
Evidently, the received power at each receiver does not depend on the other receivers' profile since $M_i\to0$, $\forall i \in \K$.

{In order to evaluate the average performance of the typical receiver over all the possible locations of the network, we assume that the receiver is uniformly distributed within the cell.} That is, $d_0$ is a random variable with probability density function (pdf) $f_{d_0}(v)=\frac{2}{\rho^2}$ \cite{kWPTRelays}. Then, the outage probability for this scenario is given in the following proposition.
\begin{prop}\label{prop1}
	The outage probability for a typical receiver uniformly distributed within a cell of radius $\rho$ is given by
	\begin{equation}\label{su}
	\mathcal{P}_o(\tau) = 1 - \frac{1}{\rho^2} \left(\frac{P \omega^2 e^2 I_0^2 x}{\tau R(r+x)^2}\right)^{1/3}.
	\end{equation}
\end{prop}

\begin{proof}
	See Appendix \ref{prf_prop1}.
\end{proof}
It is easy to obtain from \eqref{su} the minimum transmitter's output power which is required to achieve $\mathcal{P}_o(\tau) = 0$ and is
\begin{align}
P \geq \frac{\rho^6 \tau R (r+x)^2}{\omega^2 e^2 I_0^2 x}.\label{minP}
\end{align}
Clearly, with a higher mutual inductance with the transmitter, less power is required to achieve $\mathcal{P}_o(\tau)=0$.

\section{A non-cooperative game - Nash equilibrium}\label{optloads}
In this section, we consider that the mutual inductance between the transmitter and the $i$-th receiver is fixed, where $i \in \K$, while the corresponding load $x_i$ is adjustable and varies within the interval $\left[x_l,x_u\right]$. { The $i$-th receiver can adjust its own load $x_i$ aiming to maximize its harvested power $p_i$ subject to the rest $x_k$ loads, $k \neq i$. For that purpose, we define the following non-cooperative strategic game $\mathcal{G}=\{\K, x_i, u_i \}$ \cite{gametheory}}.
\begin{itemize}
	\item {\it Players}: The receiver set $\K=\{1,\ldots,K \}$.
	\item {\it Actions}: Each receiver adjusts its load resistance $x_i$, with $x_l \leq x_i\leq x_u$, in order to maximize $p_i$. 
	\item {\it Utilities:} The utility function is the harvested power by each receiver i.e., $u_i(x_i,\pmb{x}_{-i})=p_i$, where $\pmb{x}_{-i}$ denotes the strategy profile of the rest receivers.
\end{itemize}
The solution to the above non-cooperative game (if it exists), is denoted by $x_i^*$ and is the Nash equilibrium (NE), which satisfies the condition $u_i(x_i^*,\pmb{x}_{-i}^*)\geq u_i(x_i,\pmb{x}_{-i}^*)$, $\forall\; x_i \neq x_i^*$. By assuming that the other players do not change their current strategies, the best response for the $i$-th player is a strategy update rule, where the receiver selects the strategy that maximizes its individual utility function. Specifically, given a current strategy profile $\pmb{x}=(x_i,\pmb{x}_{-i})$, the strategy update of the player $i$ is the solution to the following optimization problem \vspace{-1mm}
\begin{align}
\tilde{x}_i=\arg \max_{x_l\leq x_i \leq x_u} u_i({x_i,\pmb{x}_{-i}}). 
\end{align}
The solution to the above optimization problem is given in the following proposition.
\begin{prop}\label{propp1}
The optimal strategy update for the $i$-th player given the strategy profile $\pmb{x}_{-i}$,
 \begin{equation}
 \tilde{x}_i=\begin{cases}
 \xi(\beta,\gamma), & \text{if}\; \xi(\beta,\gamma)\in [x_l, x_u],\\
  x_l, & \text{if}\; \xi(\beta,\gamma)<x_l, \\
 x_u, & \text{if}\; \xi(\beta,\gamma)>x_u,
 \end{cases}
 \end{equation}
where $\xi(\beta,\gamma)=\sqrt{r\left(r+\frac{\beta}{\gamma}\right)}$, $\beta=\omega^2 M_i^2$ and $\gamma= R+\omega^2 \sum_{k\neq i} M_k^2 (r+x_k)^{-1}$.
\end{prop}
\begin{proof}
	See Appendix \ref{ap1}.
\end{proof}
An interesting observation from the above proposition is that the best response (as well as the NE) is independent of the transmitter's output power $P$.
 \begin{remark}
 	For the special case where $x_i=x$, $M_i=M$ $\forall i \in \K$, $R\rightarrow 0$ and $r\rightarrow 0$, we have $p_i=\frac{P}{K}$. That is, the transmitter's output power is distributed symmetrically among the receivers and becomes independent of the system parameters. Furthermore, it is worth noting that when $K\rightarrow \infty$, the harvested power asymptotically converges to zero i.e., $p_i\rightarrow 0$.  
 \end{remark}
 We now investigate the loosley coupling region. We first consider the case where the $i$-th receiver is loosely coupled i.e., $M_i\rightarrow 0$. The harvested power in this case, is similar to $\eqref{s1}$ where by taking into account the effect of $\gamma$, we have
 \begin{align}
 p_i=P \frac{\omega^2 M_i^2 x_i}{\gamma(r+x_i)^2}.
 \end{align}
Then, by following Proposition \ref{propp1}, if $r \in [x_l, x_u]$, the optimal strategy update is $\tilde{x}_i=r$. Therefore, the harvested power becomes equal to $p_i=P \frac{\omega^2 M_i^2}{4\gamma r}$. In a similar way, if $M_i\rightarrow 0$, $\forall\; i \in \K$, we have $\tilde{x}_i=r$ and $p_i=P \frac{\omega_0^2 M_i^2}{4R r}$, $\forall\; i \in \K$.
\section{Numerical Results}\label{numerical}
In this section, we present numerical results in order to evaluate the performance of the considered MRC-WPT network. The constant $e$ is evaluated by $e=\frac{\pi \mu_0 N A^2 n a^2}{4}$, where $\mu_0=4 \pi \times 10^{-7}$ H/m is the magnetic permeability of the air; $N=200$ and $n=10$ correspond to the transmitter's and receivers' coils turns, respectively with corresponding radius of $A=20$ cm and $a=5$ cm \cite{Rui}. Finally, the internal resistances of the transmitter's and receivers' coils are $R=1.3440\,\Omega$ and $r=0.0672\,\Omega$ \cite{Rui}. 
\begin{figure}[t]\centering
	\includegraphics[width=0.65\linewidth]{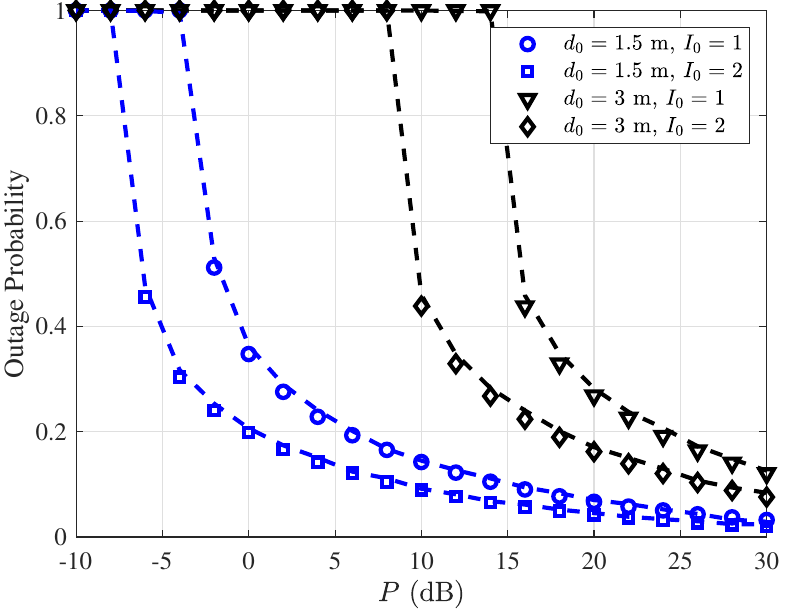}
	\vspace{-1mm}
	\caption{Outage probability in the strong coupling region; Markers and dashed lines correspond to simulations and analytical results, respectively.}\label{fig1}
	\vspace{-1mm}
\end{figure}

We first focus on performance of the typical receiver. Unless otherwise stated we consider $\rho=5$ m, $\lambda=0.1$ $\text{coils}/\text{m}^2 $, $x=2\,\Omega$ and $\tau=0.1$ W. In Fig. \ref{fig1}, we plot the outage probability with respect to the transmitter's output power for $I_0=\{1,2\}$ and $d_0=\{1.5,3\}$ m. Markers and dashed lines correspond to simulations and analytical results, respectively. {As can be seen, the proposed approximation in Proposition \ref{out}, provides a very tight upper bound for the outage probability.} Furthermore, we can see that in all cases, $\mathcal{P}_o(\tau)=1$ for $P<\frac{\tau\left(r+x\right)}{x}\left(\frac{R(r+x)}{\omega^2M_0^2}+1\right)$ (see Proposition \ref{out}), and as $P$ increases the outage probability decreases, as expected. Clearly, when the receiver is in the stronger coupling region i.e., $I_0=2$ and $d_0=1.5$ m, the lowest outage probability is achieved, as expected. When the receiver's coupling becomes weaker i.e., with lower mutual inductance e.g., with a higher $d_0$ or smaller $I_0$, the outage probability increases. Furthermore, we can see that increasing the distance from the transmitter has a significant impact on the outage probability since $M\propto d^{-3}$.

In Fig. \ref{fig2}, we consider the case where all the receivers are in the loosely coupling regime and plot the outage probability of the typical receiver for $\tau=0.1$ W. When $I_0=0.25$, we compare the outage probability by assuming a higher internal resistance for the transmitter i.e., $R=2.5\,\Omega$ which results in a worse performance for the receiver. This is expected since with higher $R$, less power is transferred to the receivers loads. For the case where $I_0=0.5$, we compare the outage probability for loads $x=\{r,1,2\}\,\Omega$. Clearly when $x=r$, the receiver achieves the lowest outage probability, since in the loosely coupling region the harvested power is maximized when $x=r$, as suggested in Section \ref{optloads}. Furthermore, from \eqref{minP}, it follows that in this case, with a transmitter's power $P\geq24.5847$ dB, we can achieve $\mathcal{P}_o(\tau)=0$ which can be validated from the figure.

\begin{figure}[t]\centering
	\includegraphics[width=0.65\linewidth]{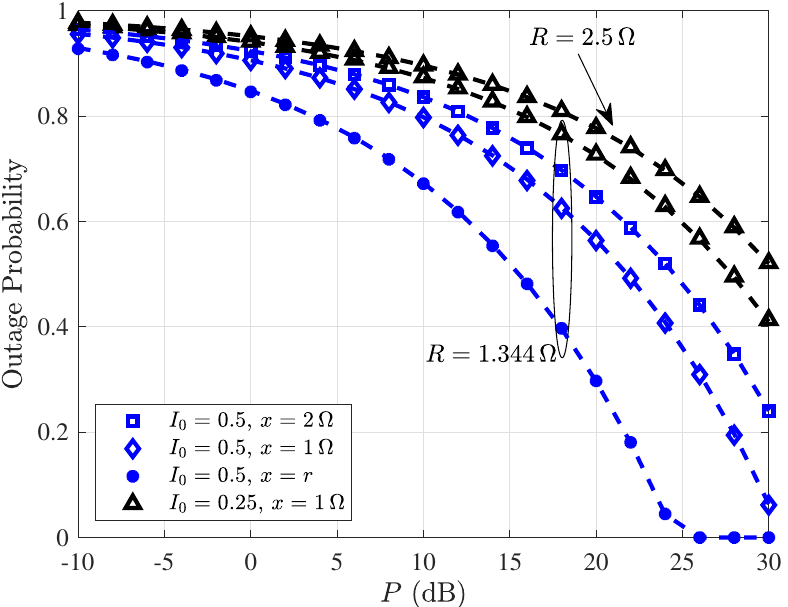}
	\vspace{-1mm}
	\caption{Outage probability in the loosely coupling region; Markers and dashed lines correspond to simulations and analytical results, respectively.}\label{fig2}
	\vspace{-1mm}
\end{figure}

We now consider a set of $K=4$ receivers with $M_1=-0.0921$ $\mu$H, $M_2=0.0402$ $\mu$H, $M_3=0.0370$ $\mu$H \cite{Rui} and $M_4=0.0245$ $\mu$H. We consider $P=10$ dB, $x_l=0.01\,\Omega$ and $x_u=5\,\Omega$ and obtain the optimal load for each receiver. That is, $x_1^*=0.1505\,\Omega$, $x_2^*=0.0796\,\Omega$, $x_3^*=0.0776\,\Omega$ and $x_4^*=0.0716\,\Omega$. We plot in Fig. \ref{game}, the harvested power achieved at each receiver and compare the cases where $x=r$ and $x=x_u$. We can see that for all the receivers, the highest harvested power is achieved with the optimal loads obtained through the non-cooperative game; the worse performance is achieved with the highest load at the receivers.
\begin{figure}[t]\centering
	\includegraphics[width=0.65\linewidth]{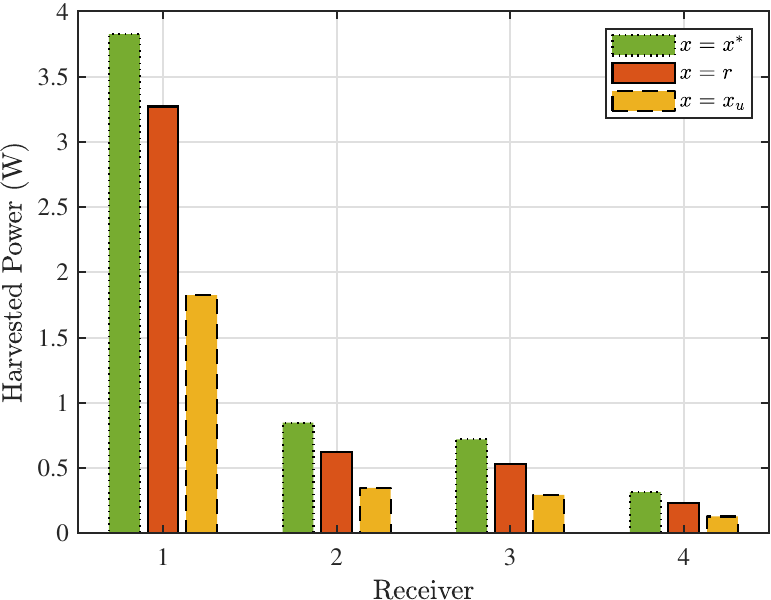}
	\vspace{-1mm}
	\caption{Harvested power at each receiver with different loads.}\label{game}
\vspace{-1mm}
\end{figure}

\section{Conclusions}
In this letter, we studied the performance of an MRC-WPT network with multiple receivers. By taking into account a minimum power threshold, we provided the outage probability in a tight approximation for the strong coupling region and in closed-form expression for the loosely coupling region, while considering spatial randomness. In order to maximize the harvested power at the receivers, we developed a non-cooperative strategic game through which the optimal load for each receiver is acquired. We presented numerical results which validated our analysis and obtained an insight overview of several system parameters which are essential for the practical applications of MRC-WPT.   

\appendix
\subsection{Proof of Lemma \ref{lemma1}}\label{prf_lemma1}
The moment generating function of $S$ is evaluated by
\begin{align}
\phi_S(t)& = \E_S\left[\exp(t S)\right]= \E_{\Psi}\left[\exp\left(t \sum_{y_i \in \Psi} I_i^2 d_i^{-6}\right)\right].
\end{align}
It is clear from \eqref{I} that $I_i \in \left[-2,2\right]$ while $\phi_S(t)$ depends on $I_i^2$. Therefore, for the sake of simplicity, we consider the substitution of $I_i$ by $\E\left[|I_i|\right]$, $\forall i \in \Psi$ and is given by
\begin{align}
\E\left[|I_i|\right]&=\E_{\theta_i}\left[\frac{1}{2\pi}\int_{0}^{2\pi}|2\sin(\theta_t)\sin(\theta_i)+\cos(\theta_t)\cos(\theta_i)|\,d\theta_t\right]\nonumber\\
&=\E_{\theta_i}\left[\frac{1}{2\pi}4 \sqrt{\cos^2(\theta_i) + 4 \sin^2(\theta_i)}\right]\nonumber\\
&=\frac{4}{\pi^2}\int_{0}^{\pi/2}\sqrt{1+3\sin^2(\theta_i)}\,d\theta_i\approx 1,
\end{align}
where the last integral constitutes of a complete elliptic integral of the second kind \cite[8.110]{GRAD} and is numerically evaluated. Then, with the approximation of $I_i\approx1$, $\phi_S(t)$ is evaluated by
{\begin{align}
\phi_S(t)& \approx\exp\left(2\pi\lambda\int_0^\rho \left(\exp\left(t v^{-6}\right)-1\right)v\,dv\right)\label{eq1}\\
&=\exp\left(2\pi\lambda\int_0^\rho \exp\left(t v^{-6}\right)v\,dv-\pi \lambda \rho^2\right),
\end{align}}where in \eqref{eq1} we make use of the probability generating functional of a PPP. The result then follows after some algebraic manipulations. 

\subsection{Proof of Proposition \ref{prop1}}\label{prf_prop1}
The outage probability of the typical receiver is evaluated as follows
\begin{align}
\mathcal{P}_o(\tau) &=\PP\left[d_0^6 > \frac{e^2I_0^2 P \omega^2 x }{\tau R (r+x)^2}\right]\label{eq2}.
\end{align}
By utilizing the pdf of $d_0$ i.e., $f_{d_0}(v)=\frac{2}{\rho^2}$, we can evaluate the cumulative distribution function by
\begin{equation}\label{CDF}
F_{d_0}(x) = \PP\left[d_0 < x\right] = \int_0^x \frac{2 v}{\rho^2} dv = \frac{x^2}{\rho^2},
\end{equation}
where $x\leq\rho$. The result then follows by applying \eqref{CDF} in \eqref{eq2}.
\subsection{Optimal strategy update} \label{ap1}
Let $\alpha=P \beta$, $\beta=\omega^2 M_i^2$ and $\gamma= R+\omega^2 \sum_{k\neq i} M_k^2 (r+x_k)^{-1}$, then we can express the harvested power at the $i$-th receiver by
\begin{equation}
f(x_i)=\frac{\alpha x_i}{\beta (r+x_i)+\gamma (r+x_i)^2}.
\end{equation}
Then, the derivative of $f(x_i)$ is given by
\begin{equation}
f'(x_i)=\frac{\alpha(\beta r +\gamma r^2 -\gamma x_i^2)}{(r+x)^2 (\beta +\gamma r+\gamma x_i)^2},
\end{equation}
which is monotonic increasing ($f'(x)>0$) in the interval $(0,\xi(\beta,\gamma))$ and monotonic decreasing ($f'(x)<0$) in the interval $(\xi(\beta, \gamma), \infty)$. Therefore, the function $f(x_i)$ has a global maximum at $f'(x_i)=0\Rightarrow x_i=\xi(\beta,\gamma)$. 

In the case where $f(x_i)$ is defined in a specific interval i.e., $\left[x_l, x_u\right]$, then $f(x_i)$ has a maximum at $\xi(\beta,\gamma)$ for $\xi(\beta,\gamma) \in \left[x_l, x_u\right]$, otherwise maximum is achieved at the boundary; at the point $x_l$ if $\xi(\beta,\gamma)<x_l$ and at $x_u$ if   $\xi(\beta,\gamma)>x_u$.

{\it Existence of the NE:} The strategy space of each player $x_i \in [x_l, x_u]$ is nonempty, convex, and compact subset of Euclidean space. In addition, the utility function $u_i(\cdot)$ is continuous in $x_i$ and from the aforementioned has a global maximum at its relative maximum point. 

{\it Uniqueness of the NE:} We prove that the function $f(x_i)$ is a standard function and therefore satisfies the positivity, monotonicity, and the extendibility. 
\begin{itemize}
\item The function $f(x_i)$ is a rational function of two positive polynomials (since $x_i>0$) and therefore $f(x_i)>0, \forall\; x_i$.
\item The function $f(x_i)$ is monotonic increasing in the interval $(0, \xi(\beta, \gamma))$ and monotonic decreasing in $(\xi(\beta, \gamma), \infty)$; therefore satisfies the requirement of monotonicity.
\item For extendibility, we need to prove that $cf(x_i)-f(cx_i)>0$ for $c>1$. That is, $cf(x_i)-f(cx_i)>0\Rightarrow \beta>-\gamma [2r+x_i(1+c)]$ which holds $\forall\; x_i$. 
\end{itemize}

\end{document}